\documentclass[pra,showpacs,twocolumn,superscriptaddress]{revtex4-1}
\usepackage{amsfonts}
\usepackage{amsmath}
\usepackage{amssymb}
\usepackage{mathtools}
\usepackage{graphicx}
\usepackage{epsfig}
\usepackage{bm}
\usepackage{enumerate}
\usepackage{multirow}
\usepackage{amsthm}
\usepackage{tabularx}
\usepackage{hyperref}

\newcommand{\ket}[1]{\left| #1 \right\rangle}

\newcommand{\bracket}[2]{\left\langle #1 | #2 \right\rangle}
\newcommand{\proj}[1]{| #1 \rangle \langle #1 |}
\newcommand{\trace}[1]{\mathrm{Tr}\left( #1 \right)}
\newcommand{\real}[1]{\mathrm{Re}\left[ #1 \right]}

\newcommand{\ip}[2]{\left\langle #1, #2 \right\rangle}
\newcommand{\sgn}{\operatorname{sgn}}

\newtheorem{lemma}{Lemma}

\makeatletter
\renewcommand*\env@matrix[1][\arraystretch]{%
  \edef\arraystretch{#1}%
  \hskip -\arraycolsep
  \let\@ifnextchar\new@ifnextchar
  \array{*\c@MaxMatrixCols c}}
\makeatother

\newcolumntype{Y}{>{\centering}X}

\begin{document}


\title{Exploring the geometry of qutrit state space using
symmetric informationally complete probabilities}




\author{Gelo Noel M. Tabia}
\email[Corresponding author: ]{gtabia@perimeterinstitute.ca}
\affiliation{Perimeter Institute for Theoretical Physics,
31 Caroline Street North, Waterloo, Ontario, Canada, N2L 2Y5}
\affiliation{Department of Physics and Astronomy and
Institute for Quantum Computing, University of Waterloo,
200 University Avenue West, Waterloo, Ontario, Canada, N2L 3G1}

\author{D. M. Appleby}
\email{mappleby@perimeterinstitute.ca}
\affiliation{Perimeter Institute for Theoretical Physics,
31 Caroline Street North, Waterloo, Ontario, Canada, N2L 2Y5}
\affiliation{Stellenbosch Institute for Advanced Study,
Wallenberg Research Centre at Stellenbosch University,
Marais Street, Stellenbosch, South Africa, 7600}

\begin{abstract}
We examine the geometric structure of qutrit state space by identifying
the outcome probabilities of symmetric informationally complete (SIC)
measurements with quantum states. We categorize the infinitely many
qutrit SICs into 8 SIC-families corresponding to independent orbits
of the extended Clifford group. Every SIC can be uniquely identified
from a set of geometric invariants that we use to establish several
properties of the convex body of qutrits, which include a simple
formula describing its extreme points, an expression for the rotation
between the probability vectors for distinct qutrit SICs, and
a polar equation for its boundary states.
\end{abstract}


\maketitle

\section{Introduction}
\label{sec:introduction}

In quantum mechanics, the state of a physical system is usually
described by a density operator, which is a positive semidefinite,
Hermitian matrix with unit trace. For any pair of Hermitian
matrices $A$ and $B$, define the Hilbert-Schmidt inner product
$
\ip{A}{B}_\mathrm{HS} = \trace{AB}.
$
The space of Hermitian matrices $\mathrm{Herm}(\mathcal{H}^d)$
on a Hilbert space $\mathcal{H}^d$ then forms a Euclidean space.
If we consider the set of density operators $\mathcal{D}(\mathcal{H}^d)$
as a subset of $\mathrm{Herm}(\mathcal{H}^d)$, then we can think of
$d$-dimensional quantum states as points in a $(d^2-1)$-dimensional
convex set $\mathcal{C} \subset \mathbb{R}^{d^2}$ that is isomorphic
to $\mathcal{D}(\mathcal{H}^d)$. We expect the geometric features
of the convex set $\mathcal{C}$ to reflect properties of
density operators. For example, the full geometry of $\mathcal{C}$ is
well-known for $d=2$; it is a solid 3-dimensional ball called the
Bloch ball. The spherical boundary of the ball corresponds to pure
states, where orthogonal states get mapped onto antipodal points, and
interior points correspond to mixtures, each of which can be
decomposed into any convex combination of pure states whose convex
hull contains that point. We can also compare how similar any two
states $\rho_1$ and $\rho_2$ are by measuring their Hilbert-Schmidt distance
$
D_\mathrm{HS}(\rho_1,\rho_2) = \sqrt{\trace{\rho_1-\rho_2}^2}.
$
Little, however, is known of the same convex geometry for quantum
states in higher dimensions. Much effort has been made
in uncovering the rich, intricate structure of $\mathcal{C}$
for $d=3$ by examining the various 2- and 3-dimensional sections
obtained from the generalized Bloch representation for qutrits
\cite{kimura2003,byrd2003,mendas2006, goyal2012,bengtsson2012} but many
details of its overall structure remain unknown.

In this paper, we analyze the geometric features of qutrits
in terms of the probabilities for a special measurement called
a symmetric informationally complete (SIC) measurement
\cite{zauner1999,caves1999,renes2004, appleby2005, weigert2006,
godsil2008, scottgrassl2010}. A SIC measurement maps each density operator into
a unique probability vector, which represents one way of specifying
an isomorphism of $\mathcal{D}(\mathcal{H}^d)$ onto a subset
of $\mathbb{R}^{d^2}$. This particular mapping allows us to
characterize quantum states as a proper subset of the
probability simplex, where the restriction is imposed mainly
by a special version of the Born rule.

We find that the SIC probabilities provide us with a novel way
of characterizing properties of qutrits, particularly with
respect to the flat geometry induced by the Euclidean metric
defined on the simplex. Here we present three main results:
\begin{enumerate}[(i)]
\item The extreme points representing pure states for qutrits are
obtained from a simple formula that picks out a submanifold of
points lying on a certain sphere.
\item The probability vectors obtained for any pair of SICs are
related by a rotation with a very simple form, which we construct
explicitly. It may be worthwhile to note that the rotation does
not necessarily arise from a unitary transformation between the
SICs involved (except when the SICs are unitarily equivalent).
\item The boundary points are described using a polar equation
that gives their radial distances from the uniform distribution,
which represents the maximally mixed state. This is different
from existing methods that analyze the geometry of qutrits by
studying the boundary of various 2-dimensional sections
\cite{kimura2003,goyal2012,bengtsson2012}.
\end{enumerate}
The motivation for such a study is twofold. Firstly, the
geometry of quantum states is interesting in its own right and a
better understanding of it may have important repercussions for
various applications of quantum information processing.
Secondly, Fuchs and Schack advocate a framework for
reformulating quantum mechanics directly in terms of
probabilities without mentioning Hilbert space at all
\cite{fuchsschack2011}. A better understanding of the structure
of SIC probabilities may prove useful in identifying the
basic axioms needed for reconstructing quantum theory
exclusively in terms of probabilities.

\section{Properties of Weyl-Heisenberg qutrit SICs}
\label{sec:sicwh}

One way to represent quantum states in terms of probabilities is to
express density operators in terms of $d^2$ linearly independent
projectors $\Pi_i = \proj{\psi_i}$ such that
\begin{equation}
\label{eq.condSIC}
| \bracket{\psi_i}{\psi_j}|^2 = \frac{d\delta_{ij} +1}{d+1}.
\end{equation}
When each projection is scaled by $\frac{1}{d}$, we get a measurement
called a symmetric informationally complete positive operator-valued
measurement (SIC-POVM), a topic of considerable interest in the
quantum physics community. In this paper, the set
$\{ \Pi_i \}_{i=1}^{d^2}$ is called a \emph{SIC} for short.

The one-to-one correspondence between the outcome probabilities
$p(i)$ of a SIC-POVM and density operators is given by
\cite{fuchs2004}
\begin{equation}
\label{eq.sicRep}
\rho = \sum_{i=1}^{d^2} \left[ (d+1)p(i)-\frac{1}{d} \right]\Pi_i.
\end{equation}
Thus, the SIC probability vectors $\vec{p}$ provide an equivalent
description of quantum states, which we call the \emph{SIC representation}.
It follows that we can always choose the coordinates of $\mathcal{C}$
such that $\vec{p} \in \mathcal{C}$ for all $\vec{p}$ associated
with a quantum state according to Eq. (\ref{eq.sicRep}).
The relation between the Hilbert-Schmidt inner product for
a pair of density operators $\rho_1$ and $\rho_2$, and the scalar
product of their respective probability vectors $\vec{p}_1$ and
$\vec{p}_2$ is given by
\begin{equation}
\trace{\rho_1\rho_2} = d(d+1)\vec{p}_1\cdot\vec{p}_2 - 1.
\end{equation}

All SICs constructed to date have a certain group covariance
property. Let $G$ be a group of $d^2$ elements and let
$g \mapsto U_g$ be a projective representation of $G$ on
$\mathcal{H}^d$. Let $\ket{\psi}\in\mathcal{H}^d$. If the set
of vectors $U_g\ket{\psi}$ generates a SIC $\mathcal{S}$, we say
that $\mathcal{S}$ is covariant with respect to $G$. The seed
vector $\ket{\psi}$ for $\mathcal{S}$ is called a fiducial vector.

In almost all known cases, the unitaries that produce
SICs belong to the Weyl-Heisenberg group. Let
$\{ \ket{j} \}_{j=0}^{d-1}$ be an orthonormal basis
for $\mathcal{H}^d$.  The Weyl-Heisenberg group is generated
by the shift $X$ and phase $Z$ operators,
\begin{align}
\nonumber
X \ket{j} &= \ket{j + 1 \mod d}, \\
Z \ket{j} &= \omega^j \ket{j}
\end{align}
where $\omega = e^{i\frac{2\pi}{d}}$.
We can act with powers of $X$ and $Z$ on a SIC fiducial $\ket{\psi}$
so that the resulting vectors
\begin{equation}
\ket{\psi_{mn}} = X^m Z^n \ket{\psi}, \qquad m,n = 0,1,\dots, d-1,
\end{equation}
satisfy Eq. (\ref{eq.condSIC}). In that case, the projectors
associated with $\ket{\psi_{mn}}$ form a SIC, which we call
a \emph{Weyl-Heisenberg SIC}.

The operators $X$ and $Z$ generate the group
\begin{equation}
W(d) = \left\{ \omega^\alpha X^m Z^n |\ \alpha,m,n = 0,1,\ldots,d-1 \right\},
\end{equation}
which we call the \emph{Weyl-Heisenberg group}.
Note that $W(d)$ is of order $d^3$. However, two unitaries
which differ only by a phase generate the same SIC projector,
so the SIC itself only contains $d^2$ elements.

The normalizer of the Weyl-Heisenberg group is called the
\emph{Clifford group} $C(d)$
\cite{gottesman1998, appleby2005,appleby2009,porteous1995},
which is itself a unitary subgroup in dimension $d$. If $U \in C(d)$ is a
Clifford unitary operator and $W(d)$ is the Weyl-Heisenberg group
then
\begin{equation}
U W(d) U^{\dag} = W(d).
\end{equation}
If the set of antiunitary operators that map $W(d)$
to itself are included, we get the \emph{extended Clifford group}.

In $d = 3$,  we have
\begin{equation}
X =
\begin{pmatrix}
0 & 0 & 1 \\
1 & 0 & 0 \\
0 & 1 & 0
\end{pmatrix},
\qquad
Z =
\begin{pmatrix}
1 & 0 & 0 \\
0 & \omega & 0 \\
0 & 0 & \omega^2
\end{pmatrix}.
\end{equation}
As a matter of convention, we label the Weyl-Heisenberg SIC
projectors $\Pi_i$ with index $i = dm + n + 1$,
so we have $i = 1,2,\ldots, 9$ for qutrit SICs. For example, $\Pi_6$ is
the SIC projector corresponding to $\ket{\psi_6} = X Z^2\ket{\psi}$.

Every Weyl-Heisenberg SIC in $d=3$ can be obtained by
acting with an (extended) Clifford (anti)unitary on a
SIC with fiducial vector
\begin{equation}
\ket{\psi_t} = \frac{1}{\sqrt{2}}
\begin{pmatrix}
0 \\
1 \\
-e^{2it}
\end{pmatrix},
\qquad t \in \left[0, \frac{\pi}{6}\right].
\end{equation}
Fiducials corresponding to distinct values of $t$ in the range
$\left[0, \frac{\pi}{6}\right]$ generate distinct orbits of the
extended Clifford group.
In Ref. \cite{appleby2005} it is shown that there are 3 types of
orbits of the extended Clifford group in $d = 3$ for which
$\ket{\psi_t}$ is in the orbit: the infinitely many generic
ones for $t \in \left(0, \frac{\pi}{6} \right)$ and 2 exceptional
ones for the endpoints $t = 0$ and $t = \frac{\pi}{6}$.

In the generic case, each extended Clifford orbit consists of 8
SICs generated by the fiducial vectors:
\begin{align}
\label{eq.fiducials}
\nonumber
\ket{\psi_t^{(0\pm)}} &= \frac{1}{\sqrt{2}}
\begin{pmatrix}
0 \\
e^{\mp it} \\
-e^{\pm it}
\end{pmatrix},
\\
\ket{\psi_t^{(\eta\pm)}} &= \sqrt{\frac{2}{3}}
\begin{pmatrix}[1.25]
\omega^{\eta} \sin t \\
\sin\left(t\pm\frac{2\pi}{3}\right) \\
\sin\left(t\mp\frac{2\pi}{3}\right)
\end{pmatrix},
\end{align}
where $\eta = 1,2,3$.

For $t = \frac{\pi}{6}$, there are 4 distinct SICs whose
fiducials can be chosen as
\begin{align}
\nonumber
\ket{\psi^{(0)}_\frac{\pi}{6}} &= \frac{1}{\sqrt{2}}
\begin{pmatrix}
0 \\
1 \\
1
\end{pmatrix},
&
\ket{\psi^{(1)}_\frac{\pi}{6}} &= \frac{1}{\sqrt{6}}
\begin{pmatrix}
\omega \\
1 \\
-2
\end{pmatrix},
\\
\ket{\psi^{(2)}_\frac{\pi}{6}} &= \frac{1}{\sqrt{6}}
\begin{pmatrix}
\omega^{2} \\
1 \\
-2
\end{pmatrix},
&
\ket{\psi^{(3)}_\frac{\pi}{6}} &= \frac{1}{\sqrt{6}}
\begin{pmatrix}
1 \\
1 \\
-2
\end{pmatrix}.
\end{align}

For $t=0$, the fiducial generating the unique SIC can be
chosen as
\begin{equation}
\label{eq.fidHesse}
\ket{\psi_0} = \frac{1}{\sqrt{2}}
\begin{pmatrix}
0 \\
1 \\
-1
\end{pmatrix}.
\end{equation}

It is worth mentioning here that the SICs of Eq. (\ref{eq.fiducials})
are inequivalent with respect to Clifford unitaries; however,
some of the SICs for different values of $t$ are still
related by a unitary operator that is not a member
of the Clifford group.
Specifically, Zhu \cite{zhu2010} has shown that the SICs
for $t$, $\frac{\pi}{9} - t$, and $\frac{\pi}{9}+t$ are, in fact,
unitarily equivalent to each other, with the unitary
transformation relating them being
\begin{equation}
U = \mathrm{diag}(1, u, u^2 ), \quad u = e^{-i \frac{2\pi}{9}},
\end{equation}
which is not a Clifford unitary.
Moreover, there are no other unitary equivalences. This means
that every pair of SICs on any two different orbits
corresponding to $t\in \left[0,\frac{\pi}{18} \right]$ are
not equivalent.

Associated with each value of $t$ are two sets of closely related
geometric quantities. The first set consists of the traces of the
product of three SIC projectors called \emph{triple products}
$T_{ijk}$,
\begin{equation}
\label{eq.tripproddef}
T_{ijk} = \trace{\Pi_{i}\Pi_{j}\Pi_{k}}.
\end{equation}
It is shown in Ref. \cite{applebyflammiafuchs2011} that two SICs are
unitarily equivalent if and only if the triple products are the
same, up to permutation.
The other set consists of the \emph{structure coefficients}
$S_{ijk}$ that describe multiplication between SIC projectors,
\begin{equation}
\label{eq.struccoefdef}
\Pi_{i}\Pi_{j} = \sum_{k}S_{ijk}\Pi_{k}.
\end{equation}
It is straightforward to show that the structure coefficients can
be obtained from the triple products in the following way:
\begin{equation}
\label{eq.structrip}
S_{ijk} = \frac{1}{d} \left[ (d+1)T_{ijk}
- \frac{d\delta_{ij} + 1}{d+1}\right].
\end{equation}
In Sec. \ref{sec:sicrep}, we shall see that the real parts
$\tilde{T}_{ijk} = \real{T_{ijk}}$ and $\tilde{S}_{ijk} =
\real{S_{ijk}}$ of the triple products and structure coefficients,
respectively, are adequate for describing probability vectors
corresponding to quantum states.

It is easy to compute $S_{ijk}$ when some of the indices are
identical:
\begin{equation}
\label{eq.strucCoefSameIndex}
S_{iii} = 1, \qquad S_{ijj} = S_{jij} = \frac{1}{4},
\qquad S_{jji} = 0.
\end{equation}
It is also straightforward to compute
$\tilde{S}_{ijk}$ for $i\neq j \neq k $ using
Eq.~(\ref{eq.tripproddef}) and Eq.~(\ref{eq.structrip}).
Taking only the real parts, the distinct nonzero values are
$-\frac{1}{4}$ and 3 other values we denote as
$x_t, y_t,$ and $z_t$:
\begin{align}
\label{eq.defxyz}
\nonumber x_t &=  -\frac{1}{6} \left( \cos 6t
+ \frac{1}{2}\right), \\
y_t &= -\frac{1}{6} \left[ \cos\left(6t + \frac{2\pi}{3} \right)
+ \frac{1}{2} \right]
\equiv x_{\frac{\pi}{9} + t},  \\
\nonumber
z_t &=  -\frac{1}{6} \left[ \cos\left(6t - \frac{2\pi}{3} \right)
 + \frac{1}{2} \right]
\equiv x_{\frac{\pi}{9} - t}.
\end{align}
It can be seen that SICs with parameter values $t, \frac{\pi}{9}-t$,
and $\frac{\pi}{9}+t$ have the same values of $x, y$, and $z$, up
to a permutation---confirming the fact mentioned earlier, that
such SICs are unitarily equivalent.

Hughston \cite{hughston2007} has shown that the SIC vectors of the
single SIC for $t=0$ can be obtained from the inflection points of
a family of cubic elliptic curves on the complex projective plane
known as the Hesse pencil (see also Bengtsson \cite{bengtsson2010}.)
There are 8 SICs with parameter value $t = \frac{\pi}{9}$ that are
unitarily equivalent to the single SIC for $t=0$, and we call
these 9 SICs the \emph{Hesse SICs}. The particular SIC specified
by Eq. (\ref{eq.fidHesse}) shall be called
the \emph{canonical Hesse SIC}.

For the canonical Hesse SIC, Eq.~(\ref{eq.defxyz}) gives
\begin{equation}
\label{eq.strucHesse}
x_0 = -\frac{1}{4}, \quad y_0 = z_0 = 0.
\end{equation}
It is the simplicity of these numbers that leads to an elegant
characterization of qutrit pure states in  Sec.~\ref{sec:sicrep}.

\begin{table}
\caption{Structure coefficient index generators for qutrit SICs.
The rules for choosing the index triples $(ijk)$ for each distinct
value of $\tilde{S}_{ijk}$ are described in the main text.}
\begin{center}
\begin{tabularx}{0.45\textwidth}{Y Y}
\hline
\hline \tabularnewline
$G_{0+} = \left[
\begin{array}{ccc}
1 & 2 & 3 \\
4 & 5 & 6 \\
7 & 8 & 9
\end{array} \right]$,
&
$G_{0-} = \left[
\begin{array}{ccc}
1 & 3 & 2 \\
4 & 6 & 5 \\
7 & 9 & 8
\end{array} \right]$, \tabularnewline \\[-1.2ex]
$G_{1+} = \left[
\begin{array}{ccc}
1 & 5 & 9 \\
2 & 6 & 7 \\
3 & 4 & 8
\end{array} \right]$,
&
$G_{1-} = \left[
\begin{array}{ccc}
1 & 9 & 5  \\
2 & 7 & 6  \\
3 & 8 & 4
\end{array} \right]$, \tabularnewline \\[-1.2ex]
$G_{2+} = \left[
\begin{array}{ccc}
1 & 6 & 8 \\
2 & 4 & 9 \\
3 & 5 & 7
\end{array} \right]$,
&
$G_{2-} = \left[
\begin{array}{ccc}
1 & 8 & 6  \\
2 & 9 & 4  \\
3 & 7 & 5
\end{array} \right]$,  \tabularnewline \\[-1.2ex]
$G_{3+} = \left[
\begin{array}{ccc}
1 & 4 & 7 \\
2 & 5 & 8 \\
3 & 6 & 9
\end{array} \right]$,
&
$G_{3-} = \left[
\begin{array}{ccc}
1 & 7 & 4 \\
2 & 8 & 5 \\
3 & 9 & 6
\end{array} \right]$.
\tabularnewline
\tabularnewline
\hline
\hline

\end{tabularx}
\end{center}
\label{tab.ig}
\end{table}

There are some simple rules for finding the index triples
corresponding to the values $-\frac{1}{4}, x_t, y_t$, and
$z_t$, which we describe next.

To each qutrit SIC-family $\ket{\psi_t^{(\eta\pm)}}~(\eta=0,1,2,3)$
in Eq.~(\ref{eq.fiducials}) we assign an index generator
$G_{\eta\pm}$ that helps us choose the index triples $(ijk)$ for
each distinct value of $\tilde{S}_{ijk}$. They are listed in
Table~\ref{tab.ig}. To illustrate what the rules are, let us
take the SIC $\ket{\psi_t^{(2+)}}$ as a specific example.

For $\tilde{S}_{ijk} = -\frac{1}{4}$, take the index triples
on the same row. Looking at $G_{2+}$ in
Table~\ref{tab.ig}, we see that the relevant set of
$(ijk)$ for $\ket{\psi_t^{(2+)}}$ is
\begin{equation}
\nonumber
\{ (168), (249), (357) \}
\end{equation}
and all permutations of indices for each $(ijk)$.

For $\tilde{S}_{ijk} = x_t$, take the index triples belonging
to the same column, or those on entirely different rows and
columns. Thus, the relevant set
of $(ijk)$ from $G_{2+}$ is
\begin{equation}
\nonumber
\{
(123), (645), (897), (147), (693), (825), (195), (627), (843)
\}
\end{equation}
and all permutations of indices for each $(ijk)$.

For $\tilde{S}_{ijk} = y_t$, take the index triples such that
the first two indices belong to the same column, and the last one
is in a different row and belongs to the succeeding column
when counting in a cyclic manner. By succeeding we mean that
``column 2 is after column 1 '', ``column 3 is after column 2'',
and ``column 1 is after column 3.'' Thus, the relevant set
of $(ijk)$ from $G_{2+}$ is
\begin{equation}
\nonumber
\{
(125), (647), (893), (236), (458), (971), (314), (569), (782)
\}
\end{equation}
and all permutations of indices for each $(ijk)$.

For $\tilde{S}_{ijk} = z_t$, we have a similar rule as in $y_t$
but take the last index from the preceding column. Thus, the
relevant set $(ijk)$ from $G_{2+}$ is
\begin{equation}
\nonumber
\{
(127), (643), (895), (238), (451), (976), (319), (562), (784)
\}
\end{equation}
and all permutations of indices for each $(ijk)$.

Any other index triple $(ijk)$ not specified above has
$\tilde{S}_{ijk} = 0$.

We do not need them in this paper but it is possible to construct
similar, though somewhat more complicated, rules for getting
the imaginary parts of $S_{ijk}$.

\section{Pure states in the SIC representation}
\label{sec:sicrep}

Since quantum state space is a compact convex body $\mathcal{C}$ in
$\mathbb{R}^{d^2}$, the Krein-Milman theorem \cite{kreinmilman1940}
states that it is equal to the convex hull of its extreme points,
which are the pure states.
It is therefore natural to ask what are the conditions on SIC
probability vectors $\vec{p}$ such that they correspond to
pure states. In terms of density operators, a pure state is
represented by a rank-1 projector, $\rho^{2}=\rho$. A remarkable
theorem \cite{flammia2004,joneslinden2005} states that for a
hermitian operator $\rho = \rho^\dag$, an equivalent
condition for defining a pure state is given by
\begin{equation}
\label{eq.pureStateThm}
\trace{\rho^{2}} = \trace{\rho^{3}} = 1.
\end{equation}

Using the SIC representation of $\rho$ given by Eq.~(\ref{eq.sicRep}),
Eq.~(\ref{eq.pureStateThm}) becomes \cite{applebydangfuchs2007}
\begin{align}
\label{eq.quadratic}
\sum_{i} p(i)^{2} &= \frac{2}{d(d+1)}, \\
\label{eq.cubictrip}
\sum_{i,j,k} T_{ijk} p(i)p(j)p(k) &= \frac{d + 7}{(d+1)^{3}}.
\end{align}
Since the right-hand side of Eq.~(\ref{eq.cubictrip}) is real, and
since the imaginary parts of the triple products are completely
antisymmetric, we have
\begin{equation}
\sum_{i,j,k} T_{ijk} p(i)p(j)p(k) = \sum_{i,j,k}
\tilde{T}_{ijk} p(i)p(j)p(k),
\end{equation}
so we may consider just the real parts $\tilde{T}_{ijk}$. We obtain
an equivalent expression for Eq.~(\ref{eq.cubictrip}) in terms of the
structure coefficients:
\begin{equation}
\label{eq.cubicstruc}
\sum_{i,j,k} \tilde{S}_{ijk} p(i)p(j)p(k) = \frac{4}{d(d+1)^{2}}.
\end{equation}

\begin{figure}
\centering
\includegraphics[scale=0.25]{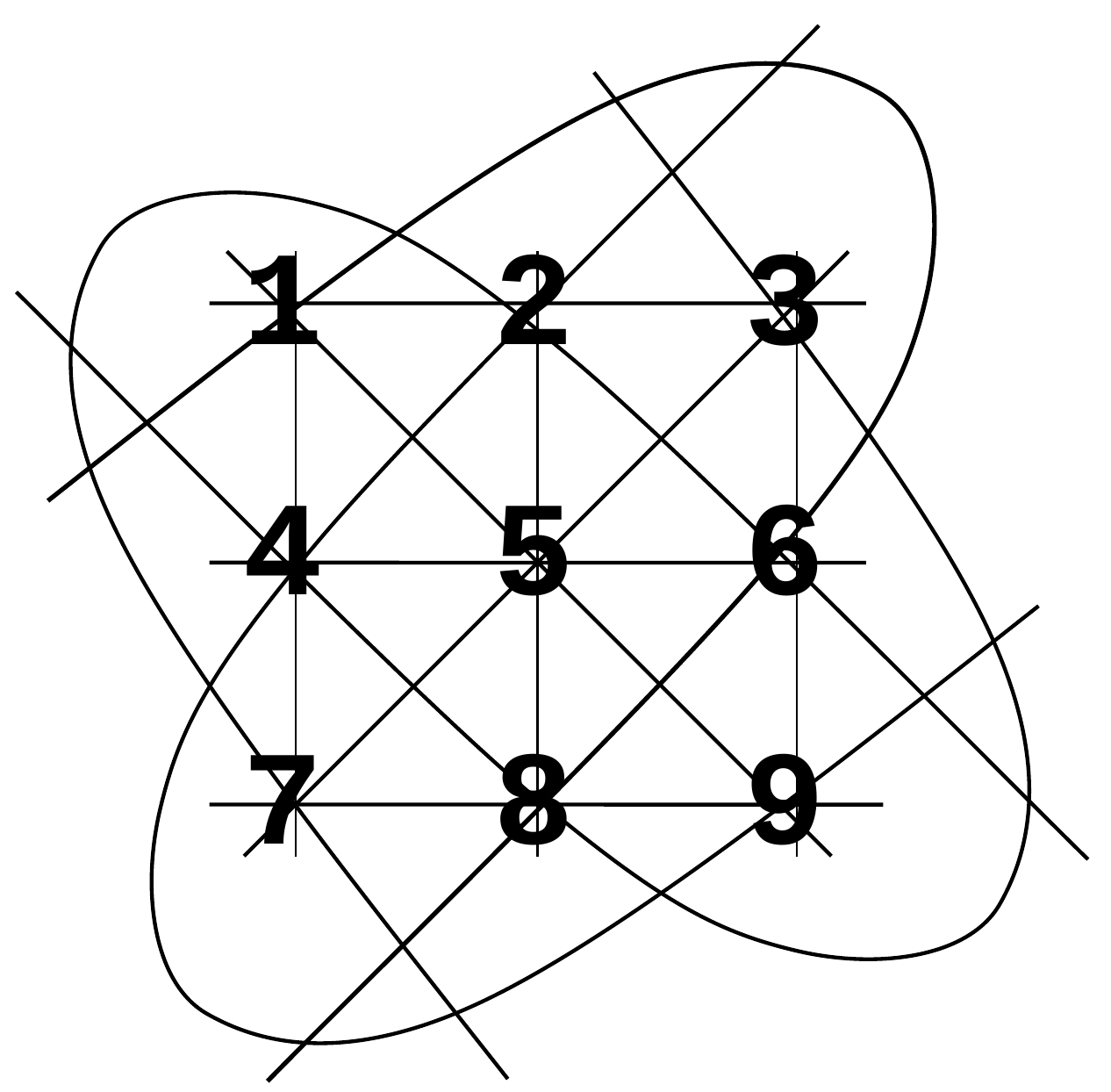}
\caption{\small \small The 12 lines of a finite affine plane over
the Galois field $\mathrm{GF(3)}$ representing the index triples $(ijk) \in Q$
in Eq.~(\ref{eq.nicequtrit2}). The indices are depicted as 9 points
and each line contains 3 points as marked.}
\label{fig.affineStriate}
\end{figure}

Specializing to the case $d = 3$, the pure states are described
by probability vectors $\vec{p}$ that satisfy
\begin{align}
\label{eq.purequtrit1}
\sum_{i} p(i)^{2} &= \frac{1}{6}, \\
\label{eq.purequtrit2}
\sum_{i,j,k} \tilde{S}_{ijk}p(i)p(j)p(k) &= \frac{1}{12}.
\end{align}
Using the results in Sec. \ref{sec:sicwh}, we find that the pure
states for the canonical Hesse SIC are given by
\begin{align}
\label{eq.nicequtrit1}
\sum_{i} p(i)^{2} &= \frac{1}{6}, \\
\label{eq.nicequtrit2}
\sum_{i} p(i)^{3} &= \frac{1}{2} \sum_{(ijk)\in Q} p(i)p(j)p(k)
\end{align}
where $Q$ is the set of index triples $(ijk)$ corresponding to
the lines drawn in Fig.~\ref{fig.affineStriate}, where
permutations of the indices $i,j,k$ are counted separately.
Interestingly, Fig.~\ref{fig.affineStriate} coincides with a
combinatorial object called a finite affine plane of order 3
(which can also be identified with the unique 2-$(9,3,1)$-design)
\cite{hirschfeld1998}. It contains 9 points and 12 lines,
and the index triples in $Q$ correspond to any 3 points on
the same line.

The pure states for any other qutrit SIC are located on the same
sphere given by Eq.~(\ref{eq.purequtrit1}) but with a different
set of values for $\tilde{S}_{ijk}$ in Eq.~(\ref{eq.cubicstruc}).
If we substitute the values in Eq.~(\ref{eq.strucCoefSameIndex})
into Eq.~(\ref{eq.cubicstruc}), we obtain
\begin{equation}
\label{eq.alternateCubic}
\frac{1}{2} \sum_{i} p(i)^3 + \sum_{i\neq j\neq k}
\tilde{S}_{ijk} p(i)p(j)p(k) = 0.
\end{equation}

\section{Orthogonal transformations between probability vectors
of distinct SICs}
\label{sec:rotprob}

In this section, we will show that, in any dimension, the probability
vectors corresponding to two different SICs are related by an
orthogonal transformation, a fact known from Ref. \cite{filippov2010}.
We will then go on to show that in dimension 3, the orthogonal
transformation takes a remarkably simple form.

In arbitrary dimension $d$, consider a pair of distinct SICs with
elements $\Pi'_i$ and $\Pi_j$. In the vector space of operators,
$\Pi'_i$ and $\Pi_j$ correspond to the vertices of 2 identical
regular simplices, which means they must be related by an orthogonal
transformation. Formally, because a SIC forms a Hermitian basis in
the space of operators, we can write
\begin{equation}
\label{eq.relatePi}
\Pi'_i = \sum_j R_{ij} \Pi_j.
\end{equation}
Since every SIC element has unit trace, taking the trace on both
sides of Eq.~(\ref{eq.relatePi}) gives us
\begin{equation}
\label{eq.rotateCond1}
\sum_{j} R_{ij} = 1.
\end{equation}
Multiplying the left-hand side of Eq.~(\ref{eq.relatePi}) by $\Pi'_j$
and the right-hand side by $\sum_l R_{jl} \Pi_l$, we have
\begin{align}
\label{eq.rotateCond2}
\nonumber \trace{\Pi'_{i}\Pi'_{j}}
&= \sum_{k,l} R_{ik}R_{jl} \left(\frac{d\delta_{kl}+1}{d+1}\right),\\
\implies \delta_{ij} &= \sum_{k} R_{ik}R_{jk}.
\end{align}
which confirms that $R_{ij}$ is indeed an orthogonal matrix.
Using Eq.~(\ref{eq.rotateCond2}), it is now straightforward
to show that
\begin{equation}
p'(i) = \sum_{j} R_{ij}p(j).
\end{equation}

Let $\Pi_i$ be the canonical Hesse SIC and let $\Pi^{(\eta\pm)}_i (t)$
be the SIC generated by the fiducial vector $\ket{\psi^{(\eta\pm)}_t}$.
Let
\begin{equation}
Q_i = \frac{1}{3}\left(4\Pi_i - I\right)
\end{equation}
be the dual basis to $\Pi_i$ (so $\mathrm{Tr}(Q_i \Pi_j) = \delta_{ij}$).
Then the orthogonal matrix which takes $\Pi_i$ onto
$\Pi^{(\eta\pm)}_i (t) $ is
\begin{equation}
R^{(\eta\pm)}_{ij}(t) = \mathrm{Tr}\left[\Pi^{(\eta\pm)}_i (t) Q_j\right]
\end{equation}
It turns out that the matrices $R^{(\eta\pm)}(t)$ have a
very simple form. In the standard two-line notation, define
the permutations
\begin{align}
\label{eq.permSICs}
p^{(0+)} &=
\begin{pmatrix}
1 & 2 & 3 & 4 & 5 & 6 & 7 & 8 & 9\\
1 & 2 & 3 & 4 & 5 & 6 & 7 & 8 & 9
\end{pmatrix}, \nonumber
\\
p^{(0-)} &=
\begin{pmatrix}
1 & 2 & 3 & 4 & 5 & 6 & 7 & 8 & 9\\
1 & 3 & 2 & 4 & 6 & 5  & 7 & 9 & 8
\end{pmatrix},
\nonumber
\\
p^{(1+)} &=
\begin{pmatrix}
1 & 2 & 3 & 4 & 5 & 6 & 7 & 8 & 9\\
1& 5& 9& 2& 6& 7& 3& 4& 8
\end{pmatrix}, \nonumber
\\
p^{(1-)} &=
\begin{pmatrix}
1 & 2 & 3 & 4 & 5 & 6 & 7 & 8 & 9\\
1& 9& 5& 2& 7& 6& 3& 8& 4
\end{pmatrix},
\nonumber
\\
p^{(2+)} &=
\begin{pmatrix}
1 & 2 & 3 & 4 & 5 & 6 & 7 & 8 & 9\\
1& 6& 8& 2& 4& 9& 3& 5& 7
\end{pmatrix}, \nonumber
\\
p^{(2-)} &=
\begin{pmatrix}
1 & 2 & 3 & 4 & 5 & 6 & 7 & 8 & 9\\
1& 8& 6& 2& 9& 4& 3& 7& 5
\end{pmatrix},
\nonumber
\\
p^{(3+)} &=
\begin{pmatrix}
 1 & 2 & 3 & 4 & 5 & 6 & 7 & 8 & 9\\
1& 4& 7& 2& 5& 8& 3& 6& 9
\end{pmatrix},
\nonumber
\\
p^{(3-)} &=
\begin{pmatrix}
 1 & 2 & 3 & 4 & 5 & 6 & 7 & 8 & 9\\
1& 7& 4& 2& 8& 5& 3& 9& 6
\end{pmatrix}.
\end{align}
Let $P^{(\eta\pm)}$ be the permutation matrix corresponding to
$p^{(\eta\pm)}$, with matrix elements
\begin{equation}
P^{(\eta\pm)}_{ij} = \delta_{j,p^{(\eta\pm)}(i)}.
\end{equation}
Also define
\begin{equation}
\label{eq.atFunc}
a(t) = \frac{1}{3}(1+2\cos2 t),
\end{equation}
\begin{equation}
A(t) =
\begin{pmatrix}
a(t) & a\left(t-\frac{\pi}{3}\right) & a\left(t+\frac{\pi}{3}\right)
\\
a\left(t+\frac{\pi}{3}\right) & a(t) & a\left(t-\frac{\pi}{3}\right)
\\
a\left(t-\frac{\pi}{3}\right)  & a\left(t+\frac{\pi}{3}\right) & a(t)
\end{pmatrix},
\end{equation}
\begin{equation}
R(t) =
\begin{pmatrix}
A(t) & 0 & 0 \\
0 & A(t) & 0 \\
0 & 0 & A(t)
\end{pmatrix}.
\end{equation}
It is then straightforward, though somewhat tedious, to verify that
\begin{equation}
R^{(\eta\pm)} (t) = \left[P^{(\eta\pm)}\right]^{-1} R(t) P^{(\eta\pm)}.
\end{equation}
Since $R(t) = I \otimes A(t)$, it follows that
\begin{equation}
\mathrm{Det}\bigl[R^{(\eta\pm)}(t)\bigr]  =
\mathrm{Det}[R(t)] = \left\{\mathrm{Det}\bigl[A(t)\bigr] \right\}^3.
\end{equation}
Because $A(t)$ is a circulant matrix, its eigenvalues are
given by
\begin{align}
\nonumber
\lambda_\ell & = a(t) + \omega^\ell a\left(t-\frac{\pi}{3}\right)
+\omega^{-\ell}a\left(t+\frac{\pi}{3}\right) \\
& = e^{2it\ell}
\end{align}
for $\ell = -1,0,1$. This implies that
$\mathrm{Det}\bigl[ A(t)\bigr]=1$.
Thus, $\mathrm{Det}\bigl[R^{(\eta\pm)} (t)\bigr] = 1$ and
$R^{(\eta\pm)} (t) $ is, in fact, a rotation matrix.

It is easily seen that
\begin{align}
R(t_1)R(t_2) & = R(t_1+t_2), &  R(0) & = I.
\end{align}
So the matrices $R(t)$ form a $1$-parameter subgroup of
the orthogonal group.

\section{The boundary of qutrit state space}
\label{sec:shapeQutrit}

A concrete way to understand the geometry of qutrit state space
is to figure out what the convex body looks like. In this regard,
we want to consider not just the pure states but all boundary points
of the set. Some valuable insight into the shape of the boundary
is gained by looking at the distance of the boundary states from the
center of the space, the maximally mixed state $\rho = \frac{1}{d}I$
as a function of direction. Specifically, we can write the SIC
probabilities in the form
\begin{equation}
\label{eq.boundaryPts}
p(i) = \frac{1}{d^2} + r n(i)
\end{equation}
where $\vec{n}$ is a direction vector with
\begin{equation}
\label{eq.directionVec}
\sum_i n(i) = 0, \qquad \sum_{i} n(i)^2 = 1.
\end{equation}
Let $r(\vec{n})$ be the value of the $r$ corresponding to the quantum
state on the boundary, which is given by Eq.~ (\ref{eq.boundaryPts}).
Here we calculate this function for the canonical Hesse SIC. Put
differently, we are looking for the polar equation describing
its boundary states.

The boundary is determined using the following lemma:
\begin{lemma}
\label{lem.boundary}
Let $\rho$ be an arbitrary Hermitian operator on a 3-dimensional
Hilbert space. Then
\begin{enumerate}[(i)]
\item $\rho$ is a density operator if and only if
\begin{equation}
\nonumber
\trace{\rho} = 1, \quad \trace{\rho^2} \leq 1, \quad
3 \trace{\rho^2}-2\trace{\rho^3} \leq 1;
\end{equation}
\item $\rho$ is a density operator for a boundary state
if and only if
\begin{equation}
\nonumber
\trace{\rho} = 1, \quad \trace{\rho^2} \leq 1, \quad
3 \trace{\rho^2}-2\trace{\rho^3} = 1;
\end{equation}
\item $\rho$ is a density operator for a pure quantum state
if and only if
\begin{equation}
\nonumber
\trace{\rho} = 1, \quad \trace{\rho^2} = 1, \quad
3 \trace{\rho^2}-2\trace{\rho^3} = 1.
\end{equation}
\end{enumerate}
\end{lemma}

\begin{proof}
We begin by proving necessity.
Suppose $\rho$ is a density matrix. It immediately follows that
$\mathrm{Tr}(\rho) = 1$ and $\mathrm{Tr}(\rho^2)  \leq 1$.
To prove the remaining inequality, let $\alpha,\beta,1-\alpha-\beta$
be the eigenvalues of $\rho$.  We find
\begin{align}
3\mathrm{Tr}(\rho^2)-2\mathrm{Tr}(\rho^3)-1
& =-6 \alpha \beta (1-\alpha -\beta)\leq 0.
\label{eq.trCond}
\end{align}
For $\rho$ to be a boundary state at least one of its eigenvalues
must vanish, in which case
\begin{align}
3\mathrm{Tr}(\rho^2)-2\mathrm{Tr}(\rho^3)-1  & =0.
\end{align}
In addition, if $\rho$ is a pure state then $\mathrm{Tr}(\rho^2)  =1$.

We now turn to the proof of sufficiency.
Let $\rho=\rho^{\dagger}$ be such that
\begin{align}
\mathrm{Tr}(\rho) & = 1, & \mathrm{Tr}(\rho^2) & \le 1,
& 3\mathrm{Tr}(\rho^2)-2\mathrm{Tr}(\rho^3) & \le 1.
\end{align}
The first equality means that we can take the eigenvalues of $\rho$
to be $\alpha,\beta,1-\alpha-\beta$.  From Eq.~(\ref{eq.trCond})
we get
\begin{equation}
\alpha \beta (1-\alpha - \beta) \ge 0.
\end{equation}
Thus, either (i) all eigenvalues are non-negative or
(ii) exactly two of them are negative.
We can show that (ii) is impossible.  Assume the contrary to hold.
Without loss of generality  $\alpha, \beta < 0$, implying that
$1-\alpha-\beta> 1$, which in turn implies $\mathrm{Tr}(\rho^2) >1$,
contrary to hypothesis.  We conclude that $\rho$ is
positive semi-definite, and consequently a density matrix.

Next  assume that
\begin{align}
\mathrm{Tr}(\rho) & = 1, & \mathrm{Tr}(\rho^2)& \le 1,
& 3\mathrm{Tr}(\rho^2)-2\mathrm{Tr}(\rho^3) & = 1.
\end{align}
Then,
\begin{equation}
\alpha \beta (1-\alpha - \beta) = 0,
\end{equation}
implying that at least one of the eigenvalues must be zero.
So $\rho$  is on the boundary of state space.

Finally assume
\begin{align}
\mathrm{Tr}(\rho) & = 1, & \mathrm{Tr}(\rho^2) & = 1,
& 3\mathrm{Tr}(\rho^2)-2\mathrm{Tr}(\rho^3) & = 1.
\end{align}
Using the argument above the eigenvalues are $0,\alpha,1-\alpha$.
Since $\mathrm{Tr}(\rho^2) = 1$ it must be that $\alpha = 0$ or $1$
and therefore $\rho$ is a rank-$1$ projection operator.
\end{proof}

We can use the lemma for the quantum states associated with the
canonical Hesse SIC. To this end, recall its structure
coefficients in Eq.~(\ref{eq.strucHesse}). We use these to
calculate $\trace{\rho^2}$ and $\trace{\rho^3}$ for $\rho$ given
by Eq.~(\ref{eq.sicRep}). We find that
\begin{align}
\label{eq.rhoTraces}
\trace{\rho^2} &= 12\sum_i p(i)^2 - 1, \\
 \nonumber \trace{\rho^3} &= 1 + 24 \sum_i p(i)^3
 - 12 \sum_{(ijk)\in Q} p(i)p(j)p(k),
\end{align}
where $Q$ is again the set of lines on the affine plane in
Fig.~\ref{fig.affineStriate}.
Substituting Eq.~(\ref{eq.boundaryPts}) into the probabilities
above, we obtain
\begin{align}
\label{eq.probSums}
\nonumber \sum_i p(i)^2 &= \frac{1}{9} + r^2, \\
\sum_i p(i)^3 &= \frac{1}{81} + \frac{r^2}{3}
 + r^3 \sum_i n(i)^3, \\
\nonumber \sum_{(ijk) \in Q} p(i)p(j)p(k) &= \frac{8}{81}
 - \frac{r^2}{3} + r^3 \sum_{(ijk)\in Q} n(i)n(j)n(k).
\end{align}
Consequently, we can restate the conditions in
Lemma \ref{lem.boundary} as follows:
\begin{enumerate}[(i)]
\item $\rho$ is a density operator if and only if
\begin{equation}
r^2 \leq \frac{1}{18}, \qquad 4 r^3 F(\vec{n})  -  r^2
+ \frac{1}{54} \geq  0;
\end{equation}
\item $\rho$ is a density operator on the boundary of the state
space if and only if
\begin{equation}
r^2 \leq \frac{1}{18}, \qquad 4 r^3 F(\vec{n}) -  r^2
 + \frac{1}{54} = 0;
\end{equation}
\item $\rho$ is a density operator for a pure state if and only if
\begin{equation}
r^2 = \frac{1}{18}, \qquad F(\vec{n}) = \frac{1}{\sqrt{2}};
\end{equation}
\end{enumerate}
where
\begin{equation}
\label{eq.functionF}
F(\vec{n}) = \sum_i n(i)^3 - \frac{1}{2} \sum_{(ijk) \in Q} n(i)n(j)n(k).
\end{equation}
Thus, the value of $r(\vec{n})$ giving the distance of a boundary
state from the completely mixed state along the direction of
$\vec{n}$ is the smallest positive root of
\begin{equation}
\label{eq.polarDistance}
4 r^3 F(\vec{n})  - r^2 + \frac{1}{54} = 0.
\end{equation}

\begin{figure}[t]
\centering
\includegraphics[scale=0.73]{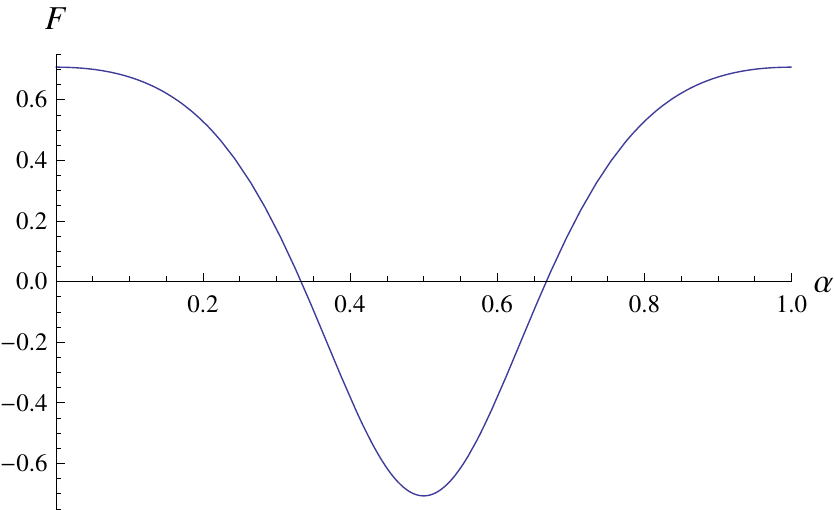}
\caption{\small A plot of $F$ as a function of one of the
eigenvalues $\alpha$ for boundary states.}
\label{fig.polarFunctionBoundary}
\end{figure}
To determine the bounds on $F \equiv F(\vec{n})$, we can write it
in terms of the eigenvalues of $\rho$ by performing some algebra
on Eq.~(\ref{eq.rhoTraces}) and Eq.~(\ref{eq.probSums}). Denoting the
eigenvalues of $\rho$ by $\alpha,\beta$, and $1- \alpha - \beta$,
we obtain
\begin{equation}
\label{eq.eigF}
F =  \frac{\sqrt{3} \left(\frac{2}{9}
- f_1 + f_2 \right)}
{\left(f_1 - \frac{1}{3}\right)^{3/2}}
\end{equation}
where $0 \leq \alpha, \beta, 1- \alpha-\beta \leq  1$ and
\begin{align}
\nonumber
f_1 &= \alpha^2 + \beta^2 + (1-\alpha-\beta)^2, \\
f_2 &= \alpha^3 + \beta^3 + (1-\alpha-\beta)^3.
\end{align}
From Eq. (\ref{eq.eigF}), it can easily be shown that
\begin{equation}
-\frac{1}{\sqrt{2}} \leq F \leq  \frac{1}{\sqrt{2}}
\end{equation}
where the upper (respectively, lower) bound is achieved when two of the
eigenvalues are identical and $<\frac{1}{3}$ (respectively, $>\frac{1}{3}$).
If all the eigenvalues are equal to $\frac{1}{3}$ this corresponds to
the maximally mixed state, for which $F$ is undefined.
For boundary states, at least one of the eigenvalues must be
zero. So the only case we need to consider for $F$ is when
$\beta = 1 - \alpha$. Figure~\ref{fig.polarFunctionBoundary} shows
$F$ as a function of $\alpha$, provided that one of the eigenvalues
vanishes.

\begin{figure}
\centering
\includegraphics[scale=0.73]{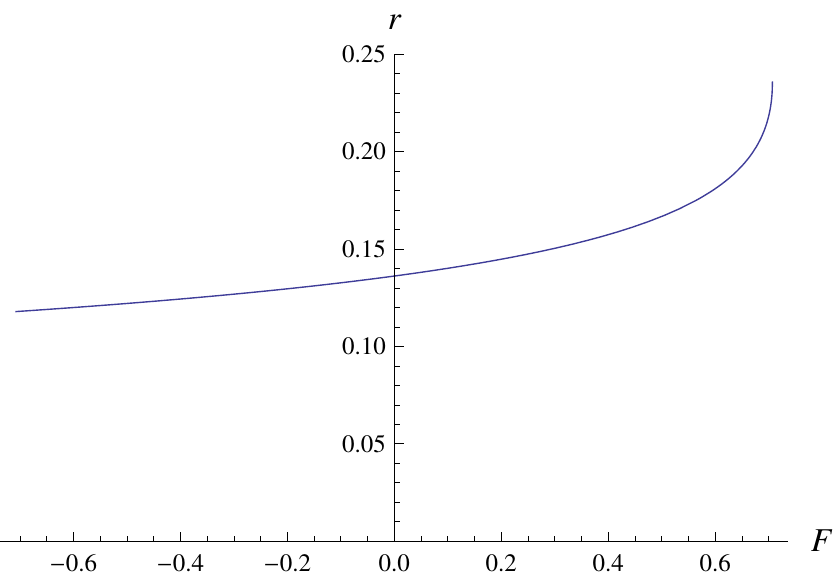}
\caption{\small The radial distance $r$ of boundary states from
the maximally mixed state, as a function of $F=F(\vec{n})$.}
\label{fig.polarRadius}
\end{figure}

In terms of $F$, the desired root in Eq.~(\ref{eq.polarDistance})
is given by
\begin{equation}
r =
\begin{dcases*}
\frac{1}{12 F} \left( 1 + \frac{g}{\omega^{\sigma}}
+ \frac{\omega^{\sigma}}{g} \right) & if $F \neq 0$, \\
\frac{1}{3\sqrt{6}} & if $F = 0$, \\
\end{dcases*}
\end{equation}
where $\sigma \equiv \sgn(F)$ and $g$ is the cube root with the
smallest positive argument in
\begin{equation}
g^3 = 1 - 4F^2 + 2F\sqrt{4F^2 - 2}.
\end{equation}
When $F = \pm\frac{1}{\sqrt{2}}$, $g = e^{i\frac{\pi}{3}}$ and
we get the bounds for $r$:
\begin{equation}
\frac{1}{6\sqrt{2}} \leq r \leq \frac{1}{3\sqrt{2}}.
\end{equation}
A plot of $r$ as a function of $F$ is shown in
Fig.~\ref{fig.polarRadius}.

Finally let us note that we can use this method to answer other
questions about the geometry of the set of qutrits.  For instance,
it has been shown \cite{aef2009, rosado2010} that the $8$-dimensional
ball $r \leq \frac{1}{3\sqrt{2}}$ is truncated by the
$7$-dimensional faces of  the probability simplex.
It is interesting to ask about the pure states located on these
faces. To answer this question, consider, for example,
the $7$-dimensional face with center point $\vec{c}$ given by
$c(i) = \frac{1}{8}$ for $i\neq 9$ and $c(9) = 0$. In this case,
we can take probability vectors of the form
\begin{equation}
p(i) = \frac{1}{8} + s m(i), \qquad
p(9) = m(9) = 0,
\end{equation}
where
\begin{equation}
\sum_{i=1}^8 m(i) = 0, \qquad
\sum_{i=1}^8 m(i)^2 = 1.
\end{equation}
Going through a similar argument as the one above, we find that
the states on the face are given by $s = 0$, or
\begin{equation}
s^2 \leq \frac{1}{24}, \qquad
F(\vec{m}) = \frac{3}{16s} \left[ 2 - \Phi(\vec{m}) \right]
\end{equation}
where
\begin{align}
\nonumber
\Phi(\vec{m}) &= \left[ m(1) + m(5) \right]^2 +
\left[ m(2) + m(4) \right]^2  \\
 &\quad +\left[ m(3) + m(6) \right]^2 +
\left[ m(7) + m(8) \right]^2
\end{align}
and $F$ is the same function defined in Eq.~(\ref{eq.functionF}).
In particular, the pure states correspond to $\vec{m}$ such that
\begin{equation}
F(\vec{m}) = \frac{3\sqrt{6}}{8}\left[2 - \Phi(\vec{m}) \right].
\end{equation}

\section{Summary and Outlook}

With the recent revival of interest in addressing foundational issues
in quantum theory, we pose a simple yet intriguing question: What is
the shape of the set of quantum states $\mathcal{C} \subset
\mathbb{R}^{d^2}$? Our preliminary attempt to address this question
revolves around a description of quantum states in terms of the
outcome probabilities of a SIC-POVM. This particular representation
allows us to exploit the intrinsic symmetry of a SIC in mapping
density operators to probability vectors, which we believe not only
serves as a natural ``coordinate system'' for studying the underlying
geometry of quantum states, but also provides us with an
interpretation of quantum states in terms of Bayesian
probabilities \cite{fuchs2010}.

In this work, we focused our attention on $d=3$, which is the simplest,
nontrivial case to examine. We considered the infinitely many
Weyl-Heisenberg qutrit SICs, which are classified into SIC-families
corresponding to orbits of the Clifford group. Each SIC can be
uniquely identified with a set of complex numbers called triple
products $T_{ijk}$, the trace of the product of three SIC elements,
whose polar angles are related to discrete geometric phases
\cite{berry1984,aharonovanandan1987} and to Bargmann invariants
\cite{bargmann1964} in complex projective space, and whose imaginary
parts give SICs the structure of a Lie algebra
\cite{applebyflammiafuchs2011}. We also have structure
coefficients $S_{ijk}$, which are the expansion coefficients when
multiplying SIC projectors, and whose real parts $\tilde{S}_{ijk}$ are
especially convenient for describing geometric properties of qutrits.

Using $\tilde{S}_{ijk}$ for the canonical Hesse SIC given by
$\ket{\psi_0}$ in Eq.~(\ref{eq.fidHesse}), we discovered the most
economical description for SIC probability vectors associated with
qutrit pure states, which are given by Eq.~(\ref{eq.nicequtrit1})
and Eq.~(\ref{eq.nicequtrit2}). Studying the probabilities for
the canonical Hesse SIC is sufficient because we demonstrated that
the probabilities for other qutrit SICs are related to it by a
9-dimensional rotation that can be expressed in terms of
a single function---$a(t)$ in Eq.~(\ref{eq.atFunc}).

The remarkable simplicity of Eq.~(\ref{eq.nicequtrit2}) suggests
that the geometric structure of qutrits is largely determined by
the symmetries associated with a finite affine plane.
For example, observe that if we consider the indices as points
in Fig.~\ref{fig.affineStriate}, the permutations given in
Eq.~(\ref{eq.permSICs}) are such that they preserve the affine
lines.
Also, using the notion of maximal consistent sets in Ref.
\cite{aef2009}, a set of probability vectors associated with the
Hesse configuration of vectors on a Hilbert space can be maximized
into a convex body that has the same largest inscribed sphere and
smallest containing sphere as qutrit state space, and also
shares some of its 2-dimensional sections. Therefore, to gain
a proper understanding of the convex geometry of qutrits,
it is crucial to understand the full significance of
Eq.~(\ref{eq.nicequtrit2}).

We also described a polar equation for the qutrit boundary.
We found that the function $F(\vec{n})$ that yields the radial
distance of a boundary state in direction $\vec{n}$ from the
uniform distribution is the same function that picks out the pure
states for the canonical Hesse SIC, i.e., $F(\vec{p}) = 0$ for
any $\vec{p}$ on the sphere containing pure states,
which again highlights the important role played by the
finite affine plane of Fig.~\ref{fig.affineStriate}.

Lastly, it is our hope that the results presented here also
serve as evidence for the utility of the SIC representation in
matters regarding quantum foundations.

\begin{acknowledgments}
The authors are grateful to Chris Fuchs and \r{A}sa Ericsson for
useful discussions. This work was supported in part by the
U.S. Office of Naval Research (Grant No.\ N00014-09-1-0247).
DMA was also supported in part by the John Templeton Foundation.
Research at Perimeter Institute is supported by the Government
of Canada through Industry Canada and by the Province of Ontario
through the Ministry of Research and Innovation.
\end{acknowledgments}


\bibliography{qutritGeom}

\end{document}